\documentclass[conference]{IEEEtran}

\usepackage{cite}
\usepackage{hyperref}

\usepackage[]{graphicx}
\graphicspath{{figures/}}
\DeclareGraphicsExtensions{.eps,.pdf}

\usepackage{amsmath,amssymb,amsthm,dsfont}
\usepackage[mathscr]{euscript}
\usepackage{mathtools}
\usepackage{verbatim}
\DeclarePairedDelimiter \norm{\lVert}{\rVert}%
\newtheorem{theorem}{Theorem}

\newtheorem{lemma}{Lemma}

\theoremstyle{remark}
\newtheorem*{remark}{Remark}
\IEEEoverridecommandlockouts

\begin{document}

\allowdisplaybreaks[4] 

\title{Strong Coordination over Noisy Channels:\\ Is Separation Sufficient?}

\author{ \thanks{This work is supported by NSF grants CCF-1440014, CCF-1439465.}
\IEEEauthorblockN{Sarah A. Obead, J\"{o}rg Kliewer}
\IEEEauthorblockA{Department of Electrical and Computer Engineering\\
New Jersey Institute of Technology\\
Newark, New Jersey 07102\\
Email: sao23@njit.edu, jkliewer@njit.edu}
\and
\IEEEauthorblockN{Badri N. Vellambi}
\IEEEauthorblockA{Research School of Computer Science\\
	Australian National University\\
	Acton, Australia 2601\\
	Email: badri.n.vellambi@ieee.org}
}

\maketitle

\begin{abstract}
	
We study the problem of strong coordination of actions of two agents 
$X$ and $Y$ that communicate over a noisy communication channel such that 
the actions follow a given joint probability distribution. We propose 
two novel schemes for this noisy strong coordination problem, and derive 
inner bounds for the underlying strong coordination capacity region. 
The first scheme is a joint coordination-channel coding scheme that 
utilizes the randomness provided by the communication channel to reduce
the local randomness required in generating the action sequence at agent $Y$.
The second scheme exploits separate coordination and channel coding where local 
randomness is extracted from the channel after decoding. Finally, 
we present an example in which the joint scheme
is able to outperform the separate scheme in terms of coordination rate.

\end{abstract}

\section{Introduction}

The problem of communication-based coordination of multi-agent systems
arises in numerous applications including mobile robotic networks, smart
traffic control, and distributed computing such as distributed games and
grid computing \cite{cuff2010coordination}. Several theoretical and applied studies on
multi-agent coordination have targeted  questions on  how agents
exchange information and how their actions can be correlated to achieve a
desired overall behavior. Two types of coordination have been addressed in
the literature -- \emph{empirical} coordination where the histogram of induced
actions is required to be close to a prescribed target distribution, and
\emph{strong} coordination, where the induced sequence of joint actions of all the
agents is required to be statistically close (i.e., nearly indistinguishable) from a chosen target
probability mass function (pmf).

Recently, the capacity regions of
several empirical and strong coordination network problems have been
established \cite{soljanin2002compressing,cuff2010coordination,
	Cuff13, SCo3TRN:2014, SCMLN:2015, bereyhi2013empirical}. Bounds
for the capacity region for the point-to-point case  were obtained in
\cite{gohari2011generating} under the assumption that the nodes communicate
in a bidirectional fashion in order to achieve coordination. A similar
framework was adopted and improved in \cite{yassaee2015channel}. In
\cite{SCo3TRN:2014,haddadpour2012coordination,bereyhi2013empirical}, the
authors addressed inner and outer bounds for the capacity region of a
three-terminal network in the presence of a relay. The work of
\cite{SCo3TRN:2014} was later extended in \cite{bloch2013strong,SCMLN:2015} to
derive a  precise characterization of the strong coordination region for
multi-hop networks. Starkly, the majority of the recent works on coordination have considered noise-free communication channels with the exception of two works: joint empirical coordination of the channel inputs/outputs of a noisy communication channel with source and reproduction sequences is considered in \cite{CS11}, and in
\cite{yassaee2016channelWithchannel}, the notion of strong coordination is
used to simulate a discrete memoryless channel via another channel.

In this work, we consider the point-to-point coordination setup illustrated in Fig.~\ref{fig:P2PCoordination}, where in
contrast to \cite{CS11} only source and reproduction sequences at two
different nodes ($X$ and $Y$) are coordinated by means of a suitable communication scheme over a discrete memoryless channel (DMC).

Specifically, we propose two different novel achievable coding schemes for
this noisy coordination scenario, and derive inner bounds to the underlying strong
capacity region. The first scheme is a joint coordination channel coding
scheme that utilizes randomness provided by the DMC to reduce the local
randomness required in generating the action sequence at Node $Y$ (see
Fig.~\ref{fig:P2PCoordination}). The second scheme exploits separate
coordination and channel coding where local randomness is extracted from the
channel after decoding. Even though the proposed joint scheme is related to the
scheme in \cite{yassaee2016channelWithchannel}, the presented scheme
exhibits a significantly different codebook construction adapted to our coordination
framework. Our scheme requires the quantification of the amount of common randomness shared by the two nodes as well as the local randomness at each of the two nodes. This is a feature that is absent from the analysis in  \cite{yassaee2016channelWithchannel}.  
Lastly, when the noisy channel and the correlation between $X$ to $Y$ are both given by binary symmetric channels (BSCs), we study
the effect of the capacity of the noisy channel on
the sum rate of common and local randomness. We conclude this work by showing that the joint scheme outperforms the
separate scheme in terms of the coordination rate in the high-capacity regime.

The remainder of the paper is organized as follows: Section~\ref{sec:notation} sets the notation. The problem of strong coordination over a noisy communication link is presented in Section~\ref{sec:problemdef}. We then derive achievability results for the noisy point-to-point coordination in Section~\ref{sec:JointScheme} for the joint scheme and in Section~\ref{sec:SepCoorRE} for the separate scheme, respectively. In Section~\ref{sec:example}, we present numerical results for both schemes when the target joint distribution is described as a doubly binary symmetric source and the noisy channel is given by a BSC.

\section{Notation}\label{sec:notation}
Throughout the paper, we denote a discrete random variable with upper-case
letters (e.g.,~$X$) and its realization with lower case letters
(e.g.,~$x$), respectively. The alphabet size of the random variable $X$ is denoted as
$|\mathcal{X}|$. We use $X^n$ to denote the finite sequence
{$[X_1,X_2,\dots,X_n]$}. 
The binary entropy function is denoted as $h_2(\cdot)$, the indicator function by $\mathds{1}(w)$, and the
counting function as
$N(\omega|w^n)=\sum_{i=1}^{n}\mathds{1}(w_i=\omega)$. 
$\mathbb P[A]$ is the
probability that the event $A$ occurs. 
The pmf of the discrete random variable $X$ is denoted as
$P_X(x)$. However, we sometime use the lower case notation (e.g.,~$p_X(x)$)
to distinguish target pmfs or alternative definitions. We let
$\mathbb{D}(P_X(x)||Q_X(x))$ denote the Kullback-Leibler divergence between
two distributions $P_X(x)$ and $Q_X(x)$ defined over an alphabet
$\cal{X}$. ${\cal T}_\epsilon^n(P_{X})$ denotes the set of
$\epsilon$-strongly letter-typical sequences of length $n$.
Finally, $P^{\otimes  n}_{X_1X_2\dots X_k}$ denotes the joint pmf of $n$
i.i.d.~random variables $X_1,X_2,\dots, X_k$.

\section{Problem Definition} \label{sec:problemdef}

\begin{figure}[t!]
	\centering
	{\includegraphics[scale=0.575]{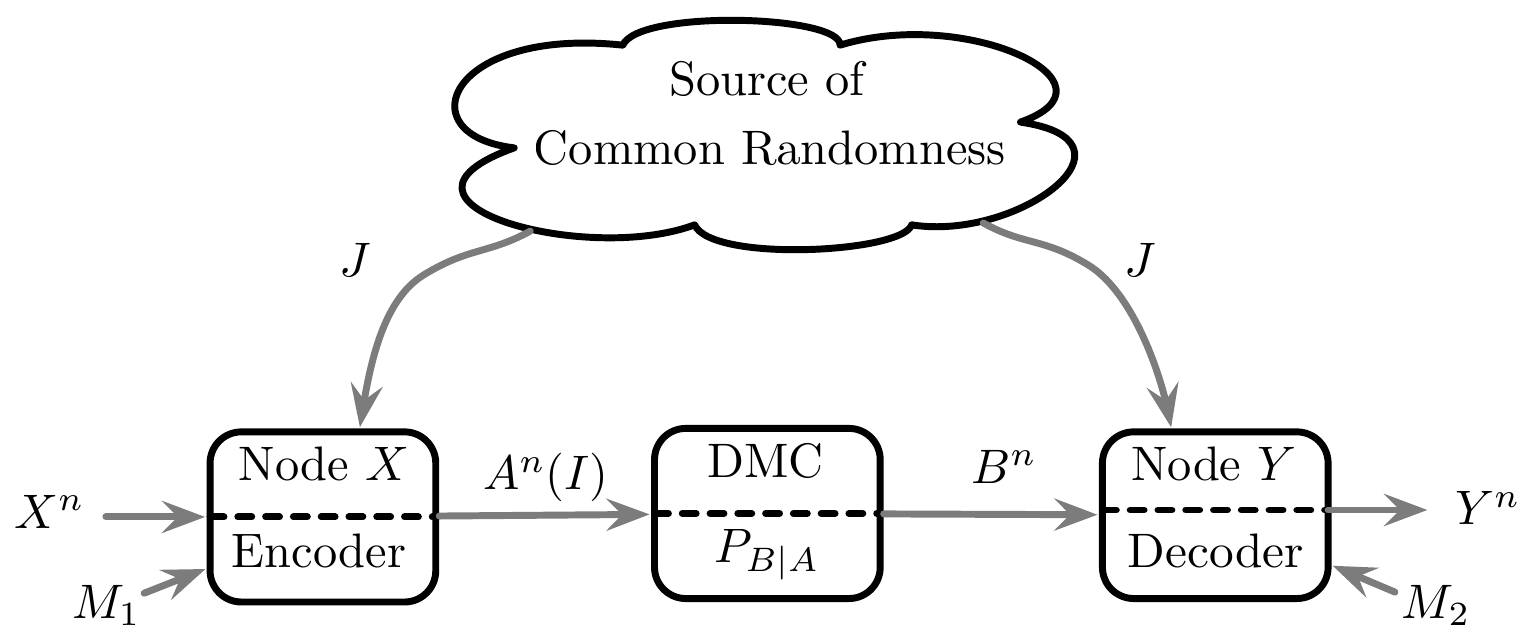}}
	\caption{Point-to-point strong coordination over a DMC.}
	\label{fig:P2PCoordination}
\end{figure}

The point-to-point coordination setup we consider in this work is depicted
in Fig.~\ref{fig:P2PCoordination}. Node $X$ receives a sequence of actions
$X^n \in \mathcal{X}^n$ specified by nature where $X^n$ is i.i.d.~according
to a pmf $p_X$. Both nodes have access to shared randomness $J$ at rate
$R_o$ bits/action from a common source, and each node possesses local
randomness $M_k$ at rate $\rho_k$, $k=1,2$. Thus, in designing a \emph{block} scheme to coordinate $n$ actions of the nodes, we assume $J\in \{1,\ldots, 2^{nR_o}\}$, and $M_k\in \{1,\ldots, 2^{n\rho_k}\}$, $k=1,2$, and we wish to communicate a
codeword $A^n(I)$ over the rate-limited DMC $P_{B|A}(b|a)$ to Node $Y$,
where $I$ denotes the (appropriately selected) coordination message. The \emph{codeword} $A^n(I)$ is
constructed based on the input action sequence $X^n$, the local randomness
$M_1$ at Node $X$, and the common randomness $J$. Node $Y$ generates a
sequence of actions $Y^n\in \mathcal{Y}^n$ based on the received codeword
$B^n$, common randomness $J$, and local randomness $M_2$. We assume that the common randomness is independent
of the action specified at Node $X$. A tuple $(R_o, \rho_1,\rho_2)$ is deemed \emph{achievable} if for each $\epsilon>0$, there exist $n\in\mathbb{N}$ and a (strong coordination) coding scheme such that the joint
pmf of actions $\hat{P}_{X^n,Y^n}$ induced by this scheme and the $n$-fold product\footnote{This is the joint pmf of $n$ i.i.d. copies of $(X,Y)\sim p_{XY}$.} of the desired joint pmf  
$P^{\otimes n}_{XY}$ are \emph{close} in total variation, i.e.,
\begin{equation}\label{eq:StrngCoorCondtion}
\norm{\hat{P}_{X^nY^n}-P^{\otimes n}_{XY}}_{{\scriptscriptstyle TV}} 
<\epsilon.
\end{equation}

We now present the two achievable coordination schemes.

\section{Joint Coordination Channel Coding}\label{sec:JointScheme}
This scheme follows an approach similar to those in \cite{cuff2010coordination,bloch2013strong,SCo3TRN:2014,SCMLN:2015} where
coordination codes are designed based on allied channel resolvability
problems~\cite{han93:_approx}. The structure of the allied problem pertinent to the coordination problem at hand is given in
Fig.~\ref{fig:StrongCoordinationAllied}. 
The aim of the allied problem is to generate $n$ symbols for two correlated sources
$X^n$ and $Y^n$ whose joint statistics is close to $P^{\otimes n}_{XY}$ as defined by \eqref{eq:StrngCoorCondtion}. To do so, we employ three independent
and uniformly distributed messages $I$, $K$, and $J$ and two codebooks $\mathscr{A}$
and $\mathscr{C}$ as shown in Fig.~\ref{fig:StrongCoordinationAllied}. To define the two
codebooks, consider auxiliary random variables $A\in\mathcal{A}$ and $C\in\mathcal{C}$ jointly correlated with $(X,Y)$ as  $P_{XYABC}=P_{AC}P_{X|AC}P_{B|A}P_{Y|BC}$.

From this factorization it can be seen that the scheme consists of two 
\emph{reverse test} channels  $P_{X|AC}$ and $P_{Y|AC}$ used to generate the sources from the
codebooks.
In particular, 
$P_{Y|AC}=P_{B|A}P_{Y|BC}$, i.e., the randomness of the DMC contributes to the
randomized generation of $Y^n$.

Generating $X^n$ and $Y^n$ from $I$, $K$, $J$ represents a complex channel
resolvability problem with the following ingredients:

\begin{itemize}
	\item Nested codebooks: Codebook $\mathscr{C}$ of size $2^{n(R_o+R_c)}$ is generated
		i.i.d.~according to pmf $P_{C}$, i.e.,
		$C^n_{ij}\sim P_{C}^{\otimes n}$ for all $(i,j) \in \cal{I}\times \cal{J}$.	Codebook $\mathscr{A}$ is generated by randomly selecting $A^n_{ijk}\sim  P_{A|C}^{\otimes n}(\cdot|{C_{ij}^n})$ for all $(i,j,k) \in \cal{I}\times\cal{J}\times\cal{K}$. 
	\item Encoding functions:\\
		$C^n\!: \{1,2,\dots,2^{nR_c}\}\!\times\! \{1,2,\dots,2^{nR_o}\}\! \rightarrow \mathcal{C}^n$,\\
		$A^n\!: \{1,\dots,2^{nR_c}\}\!\times\! \{1,\dots,2^{nR_o}\}\!\times\! \{1,\dots,2^{nR_a}\}\! \rightarrow \mathcal{A}^n$.
	\item Indices: $I,J,K$ are independent and uniformly distributed over $\{1,\dots,2^{nR_c}\}$,
		$\{1,\dots,2^{nR_o}\}$, and $\{1,\dots,2^{nR_a}\}$, respectively. 
		These indices select the pair of codewords $C^n_{IJ}$ and $A^n_{IJK}$ from codebooks $\mathscr{C}$ and $\mathscr{A}$.
	\item The selected codewords $C^n_{IJ}$ and $A^n_{IJK}$ are then
		passed through DMC $ P_{X|AC}$ at Node
		$X$, while at Node $Y$, codeword $A^n_{IJK}$ is sent
		through DMC $P_{B|A}$ whose output $B^n$ is used to decode codeword $C^n_{\hat{I}J}$ and both are then passed
		through DMC $P_{Y|BC}$ to obtain $Y^n$. 
\end{itemize} 

Since the codewords are randomly chosen, the induced joint~pmf of the generated actions and codeword indices in the allied problem is itself a random variable and depends on the random codebook. Given a realization of the codebooks 
\begin{align}
\mathsf C \triangleq (\mathscr{A},\mathscr{C})=\left\{ a_{ijk}^n, c_{ij}^n: \substack{i\in\{1,\ldots, 2^{nR_c}\}\\j\in\{1,\ldots, 2^{nR_o}\}\\ k\in\{1,\ldots, 2^{nR_a}\}}\right\}, \label{eqn-codedefn}
\end{align}
 the code-induced joint~pmf of the actions and codeword indices in the allied problem is given by
\begin{multline}\label{eqn-codeinducedAct}
\hspace{-2mm}\mathring{P}_{X^nY^nIJK}(x^n,y^n,i,j,k) \triangleq \frac{P_{X|AC}^{\otimes n}(x^n|a_{ijk}^n,c^n_{ij})}{2^{n(R_c+R_o+R_a)}}\\
\hspace{-1.5mm}\times  \Big(\sum_{b^n,\hat i}  P_{B|A}^{\otimes n}(b^n|a_{ijk}^n) \mathsf P_{\hat I | B^nJ}(\hat i|b^n\hspace{-0.75mm},j) P_{Y|BC}^{\otimes n}(y^n|b^n,c^n_{\hat ij}) \Big),
\end{multline}	

\noindent where $\mathsf P_{\hat I|B^nJ}$ denotes the pmf induced by the operation of decoding the index $I$ using the common randomness and the channel output at Node $Y$. Note that by denoting the decoding operation as a pmf, we can even incorporate randomized decoders. Note also that the indices for the $C$-codeword that generate $X$ and $Y$ sequences  in \eqref{eqn-codeinducedAct} can be different since the decoding of the index $I$ at Node $Y$ may fail. We are done if we accomplish the following tasks: (1) identify conditions on $R_o, R_c,R_a$ under which the code-induced pmf $\mathring P_{X^nY^n}$ is \emph{close} to the design pmf $P_{XY}^{\otimes n}$ in the total variation sense; and (2) devise a 
strong coordination scheme by inverting the operation at  Node $X$. This will be done in following sections by subdividing the analysis of the allied problem.

\begin{figure}[t!]
	\centering
	{\includegraphics[scale=0.625]{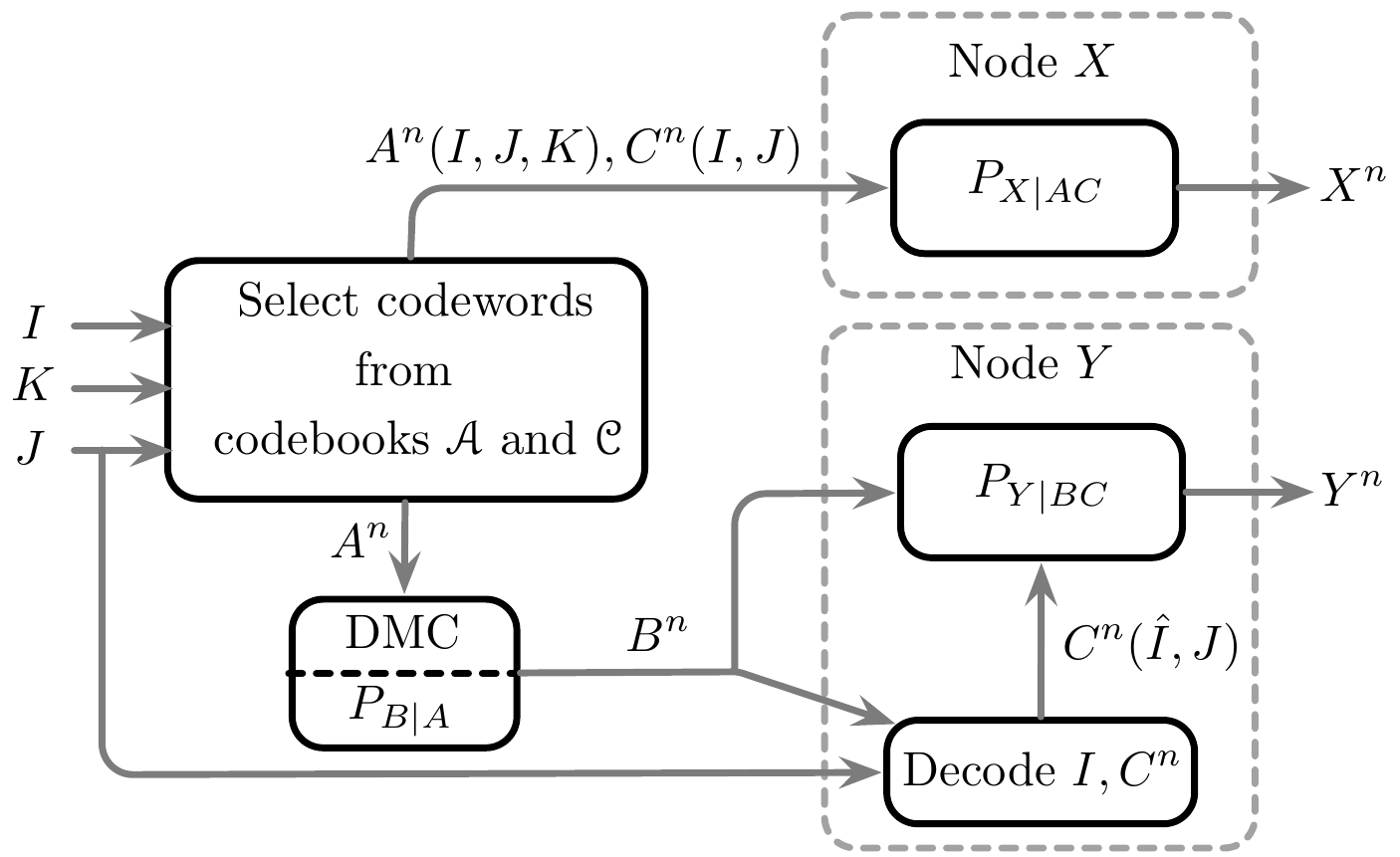}}
	\vspace{-1ex}
	\caption{A joint scheme for the allied problem.} 
	\vspace{-2ex}
	\label{fig:StrongCoordinationAllied}
\end{figure}

\subsection{Resolvability constraints} \label{sec:ChResCons}
Assuming that the decoding of $I$ and the codeword $C^n_{IJ}$ occurs perfectly at Node $Y$, we see that the code-induced joint pmf induced by the scheme for the allied problem for a given realization of the codebook $\mathsf C$ in \eqref{eqn-codedefn} is
\begin{multline}\label{eqn-codeinduced}
\check{P}_{X^nY^nIJK}(x^n,y^n,i,j,k)=\frac{P_{X|AC}^{\otimes n}(x^n|a_{ijk}^n,c^n_{ij})}{2^{n(R_c+R_o+R_a)}}\\
\times \Big( \sum_{b^n} P_{B|A}^{\otimes n}(b^n|a_{ijk}^n)  P_{Y|BC}^{\otimes n}(y^n|b^n,c^n_{ij}) \Big).
\end{multline}

The following result quantifies when the induced distribution in \eqref{eqn-codeinduced} is close to the $n$-fold product of the design pmf $P_{XY}$.  

\begin{lemma}[Resolvability constraints]
	\label{Lma:Reslv}
	The total variation between the code-induced pmf $\check{P}_{X^nY^n}$ in \eqref{eqn-codeinduced} and the desired pmf $P^{\otimes n}_{XY}$ asymptotically vanishes, i.e., 
	${\mathbb E_\mathsf{C}}\big[ \norm*{\check{P}_{X^nY^n}-P^{\otimes n}_{XY}}_{{\scriptscriptstyle TV}}\big] \rightarrow 0$ as $n\rightarrow \infty$, if 
	\begin{align}
	R_a+R_o+R_c& > I(XY;AC),\label{eqn-resolve1}\\
	 R_o+R_c &> I(XY;C).\label{eqn-resolve2}
	\end{align}
	 Note that here $\mathbb E_\mathsf{C}$ denotes the expectation over the random realization of the codebooks.
\end{lemma}
\begin{proof}
{ In the following}, we drop the subscripts from the pmfs for simplicity. 
Let $R\triangleq R_a+R_c+R_o$,  and choose $\epsilon>0$. Consider the argument for $\mathbb E_\mathsf{C}\big[\mathbb{D}(\check{P}_{X^nY^n}||P^{\otimes n}_{XY})\big]$ shown at the top of the following page.

	\begin{figure*}[!t]
	{\fontsize{10pt}{12pt} 
	\begin{equation*}\label{eqn_dbl_y}
	\begin{split} 
	&\mathbb E_\mathsf{C}\big[\mathbb{D}(\check{P}_{X^nY^n}||P^{\otimes n}_{XY})\big]\\
	&=\mathbb E_\mathsf{C}\Bigg[\sum_{x^n,y^n} \Big( \sum_{i,j,k} \dfrac{P(x^n|A^{n}_{ijk},C^{n}_{ij})P(y^n|A^{n}_{ijk},C^{n}_{ij})}{2^{nR}}\Big)
		\log \Bigg( \sum_{i',j',k'} \dfrac{P(x^n|A^{n}_{i'j'k'},C^{n}_{i'j'})P(y^n|A^{n}_{i'j'k'},C^{n}_{i'j'})}{2^{nR}P^{\otimes n}_{XY}(x^n,y^n)}\Bigg)\Bigg] \\
	&\stackrel{(a)}{=}\! \sum_{x^n,y^n} \! \mathbb E_{A^{n}_{ijk}C^{n}_{ij}} \Bigg[\! \Big( \sum_{i,j,k} 		\dfrac{P(x^n|A^{n}_{ijk},C^{n}_{ij})P(y^n|A^{n}_{ijk},C^{n}_{ij})}{2^{nR}}\Big) 
		 \mathbb E_{\mathrm{rest}} \Big[\! \log \Big(\!\! \sum_{i',j',k'} \!\!\dfrac{P(x^n|A^{n}_{i'j'k'},\!C^{n}_{i'j'})P(y^n|A^{n}_{i'j'k'},\!C^{n}_{i'j'})}{2^{nR}P^{\otimes n}_{XY}(x^n,y^n)}\Big)\! \Big| A^{n}_{ijk}C^{n}_{ij}\Big]\!\Bigg] \\ 
	&\stackrel{(b)}{\leq}\! \sum_{x^n,y^n} \! \mathbb E_{A^{n}_{ijk}C^{n}_{ij}}\Bigg[\! \Big( \sum_{i,j,k} \dfrac{P(x^n|A^{n}_{ijk},C^{n}_{ij})P(y^n|A^{n}_{ijk},C^{n}_{ij})}{2^{nR}}\Big)
	\log \Big(\mathbb E_{\mathrm{rest}} \Big[\!\!\sum_{i',j',k'}\! \!\dfrac{P(x^n|A^{n}_{i'j'k'},\!C^{n}_{i'j'})P(y^n|A^{n}_{i'j'k'},\!C^{n}_{i'j'})}{2^{nR}P^{\otimes n}_{XY}(x^n,y^n)}\! \Big| A^{n}_{ijk}C^{n}_{ij}\Big]\! \Big)\!\Bigg] \\	
	& \stackrel{(c)}{=}\! \sum_{x^n,y^n} \sum_{a^{n}_{ijk},c^{n}_{ij}} \sum_{i,j,k} \dfrac{P(x^n,y^n,a^{n}_{ijk},c^{n}_{ij})}{2^{nR}}\log \Bigg( \sum_{\substack{i',j',k':\\(i',j',k')=(i,j,k)}} \mathbb E_{A^{n}_{ijk}C^{n}_{ij}} \Big[ \dfrac{P(x^n|A^{n}_{i'j'k'},C^{n}_{i'j'})P(y^n|A^{n}_{i'j'k'},C^{n}_{i'j'})}{2^{nR}P^{\otimes n}_{XY}(x^n,y^n)}\Big|A^{n}_{ijk}C^{n}_{ij}\Big]\\
	&\hspace{2.5in} +\sum_{\substack{i',j',k':\\(i',j')=(i,j),(k'\neq k)}} \mathbb E_{A^{n}_{ijk}C^{n}_{ij}} \Big[ \dfrac{P(x^n|A^{n}_{i'j'k'},C^{n}_{i'j'})P(y^n|A^{n}_{i'j'k'},C^{n}_{i'j'})}{2^{nR}P^{\otimes n}_{XY}(x^n,y^n)}\Big|A^{n}_{ijk}C^{n}_{ij}\Big]\\ 
	&\hspace{2.5in} +\sum_{\substack{i',j',k'':\\(i',j')\neq(i,j)}} \mathbb E_{A^{n}_{ijk}C^{n}_{ij}} \Big[ \dfrac{P(x^n|A^{n}_{i'j'k'},C^{n}_{i'j'})P(y^n|A^{n}_{i'j'k'},C^{n}_{i'j'})}{2^{nR}P^{\otimes n}_{XY}(x^n,y^n)}\Big|A^{n}_{ijk}C^{n}_{ij}\Big] \Bigg) \\
	& \stackrel{(d)}{=} \sum_{x^n,y^n} \sum_{a^{n}_{ijk},c^{n}_{ij}} \sum_{i,j,k} \dfrac{P(x^n,y^n,a^{n}_{ijk},c^{n}_{ij})}{2^{nR}}\log \Bigg(\dfrac{P(x^n,y^n|a^{n}_{ijk},c^{n}_{ij})}{2^{nR}P^{\otimes n}_{XY}(x^n,y^n)} +\sum_{\substack{i',j',k':\\(i',j')=(i,j),(k'\neq k)}} \dfrac{P(x^n,y^n|c^{n}_{ij})}{2^{nR}P^{\otimes n}_{XY}(x^n,y^n)} \\
	&\hspace{2.6in} +\sum_{\substack{i',j',k':\\(i',j')\neq(i,j)}} \dfrac{P^{\otimes n}_{XY}(x^n,y^n)}{2^{nR}P^{\otimes n}_{XY}(x^n,y^n)} \Bigg) \\
	&\stackrel{(e)}{\leq}\sum_{x^n,y^n,a^{n}_{ijk},c^{n}_{ij}} P(x^n,y^n,a^{n}_{ijk},c^{n}_{ij}) \log \Bigg (\dfrac {P(x^n,y^n|a^{n}_{ijk},c^{n}_{ij})}{2^{nR}P^{\otimes n}_{XY}(x^n,y^n)} +(2^{R_a}) \dfrac{P(x^n,y^n|c^{n}_{ij})}{2^{nR}P^{\otimes n}_{XY}(x^n,y^n)} +1 \Bigg) \\
	&\stackrel{(f)}{\leq} \Bigg[ \sum_{\substack{ x^n,y^n,a^{n}_{ijk},c^{n}_{ij}:\\ (x^n,y^n,a^n,c^n)\in {\cal T}_\epsilon^n( p_{XYAC})} } P(x^n,y^n,a^{n}_{ijk},c^{n}_{ij}) \log \Bigg (\dfrac{2^{-nH(XY|AC)(1-\epsilon)}} {2^{nR}2^{-nH(XY)(1+\epsilon)}} + \dfrac{2^{-nH(XY|C)(1-\epsilon)}}{2^{n(R_o+R_c)}2^{-nH(XY)(1+\epsilon)}} + 1 \Bigg)\Bigg]\\ 
	&\hspace{2.4in} + \mathbb P\big((x^n, y^n,a^{n}_{ijk},c^{n}_{ij}) \notin {\cal T}_\epsilon^n( p_{XYAC}) \big)\log(2\mu_{XY}^{-n}+1)\\
	&\stackrel{(g)}{\leq}\Bigg[ \sum_{\substack{ x^n,y^n,a^{n}_{ijk},c^{n}_{ij}:\\ (x^n,y^n,a^n,c^n)\in {\cal T}_\epsilon^n( p_{XYAC})}} P(x^n,y^n,a^{n}_{ijk},c^{n}_{ij}) \log \Bigg(\dfrac{2^{n(I(XY;AC)+\delta(\epsilon))}}{2^{nR}}+ \dfrac{2^{n(I(XY;C)+\delta(\epsilon))}}{2^{n(R_o+R_c)}}+1 \Bigg) \Bigg]\\ 
	&\hspace{2.4in} + \big(2{\cal |X||Y||A||C|}e^{-n\epsilon^2\mu_{XYAC}}\big)\log(2\mu_{XY}^{-n}+1) 
	\xrightarrow{n\rightarrow\infty} 0.
	\end{split}
	\end{equation*} }
	\hrulefill
	\vspace*{4pt}
\end{figure*}

\noindent In this argument:
\begin{itemize} 
\item[(a)] follows from the law of iterated expectation. Note that we have used $(a^{n}_{ijk},c^{n}_{ij})$ to denote the codewords corresponding to the indices $(i,j,k)$, and $(a^{n}_{i'j'k'},c^{n}_{i'j'})$ to denote the codewords corresponding to the indices $(i',j',k')$, respectively.
\item[(b)] follows from Jensen's inequality.
\item[(c)] follows from dividing the inner summation over the indices $(i',j',k')$ into three subsets based on the indices $(i,j,k)$ from the outer summation.
\item[(d)] follows from taking the expectation within the subsets in (c) such that when
\begin{itemize}
	\item $(i',j')=(i,j),(k'\neq k)$: $a^{n}_{i'j'k'}$ is conditionally independent of $a^{n}_{ijk}$ following the nature of the codebook construction (i.e., i.i.d.~at random);
	\item $(i',j')\neq(i,j)$: both codewords ($a^{n}_{ijk},c^{n}_{ij}$) are independent of $(a^{n}_{i'j'k'},c^{n}_{i'j'})$ regardless of the value of $k$. As a result, the expected value of the induced distribution with respect to the input codebooks is the desired distribution $P^{\otimes n}_{XY}$ \cite{cuff2010coordination}.
\end{itemize} 
\item[(e)] follows from
\begin{itemize}
	\item $(i',j',k')=(i,j,k)$: there is only one pair of codewords $(a^{n}_{ijk},c^{n}_{ij})$;
	\item when $(k'\neq k)$ while $(i',j')=(i,j)$ there are $(2^{nR_a}-1)$ indices in the sum; 
	\item $(i',j')\neq(i,j)$: the number of the indices is at most $2^{nR}.$	
\end{itemize} 
\item[(f)] results from splitting the outer summation: The first summation contains typical sequences and is bounded by using the probabilities of the typical set. The second summation contains the tuple of sequences when the pair of actions sequences $x^n, y^n$ and codewords $c^n,a^n$ are not $\epsilon$-jointly typical (i.e., $(x^n,y^n,a^n,c^n)\notin {\cal T}_\epsilon^n( P_{XYAC})$). This sum is upper bounded following \cite{SCo3TRN:2014} with $\mu_{XY} =\min_{x,y} \big(P_{XY}(x,y)\big)$.
\item[(g)] following the Chernoff bound of the probability that a sequence is not strongly typical \cite{kramer2008TinMUIT} where $\mu_{XYAC} =\min_{x,y,a,c} \big(P_{XYAC}(x,y,a,c)\big)$.\\
\end{itemize}
Consequently, the contribution of typical sequences can be made asymptotically small if
\[
	R_a+R_o+R_c >  I(XY;AC),\quad
	R_o+R_c  >  I(XY;C),
\]
while the second term converges to zero exponentially fast with $n$  \cite{kramer2008TinMUIT}. Finally, by applying Pinsker's inequality we have  
\begin{align} \label{eq:resolvProof}
	\mathbb E_\mathsf{C}\big[||\check{P}_{X^nY^n}&-P^{\otimes n}_{XY}||_{{\scriptscriptstyle TV}}\big] 
	 \leq  \mathbb E_\mathsf{C}\Big[\sqrt{2\mathbb{D}(\check{P}_{X^nY^n}||P^{\otimes n}_{XY})} \;\Big] \notag\\
	& \leq \sqrt{2\mathbb E_\mathsf{C} \big[\mathbb{D}(\check{P}_{X^nY^n}||P^{\otimes n}_{XY})\big]}\mathop{\longrightarrow}^{n\rightarrow \infty} 0.
\end{align}
\end{proof}

\begin{remark}
Given $\epsilon>0$, $R_a$, $R_o$, $R_c$ satisfying \eqref{eqn-resolve1} and \eqref{eqn-resolve2}, it follows from \eqref{eq:resolvProof} that there exist an $n\in\mathbb N$ and a random codebook realization for which the code-induced pmf between the indices and the pair of actions satisfies
\begin{align}
||\check{P}_{X^nY^n}-P^{\otimes n}_{XY}||_{{\scriptscriptstyle TV}}  <\epsilon.
\end{align}
\end{remark}

\subsection{Decodability constraint} \label{sec:DecoCons}
Since the operation at Node $Y$ in Fig.~\ref{fig:StrongCoordinationAllied} involves the decoding of $I$ and thus the codeword $C^n(I,J)$ using $B^n$ and $J$, the induced distribution of the scheme for the allied problem will not match that of \eqref{eqn-codeinduced} unless and until we ensure that the decoding succeeds with high probability as $n\rightarrow \infty$. The following lemma quantifies the necessary rate for this decoding to succeed asymptotically almost always.

\begin{lemma}[Decodability constraint]
\label{Lma:Decod}
Let $\hat I, C^n_{\hat IJ}$ be the output of a typicality-based decoder that uses common randomness $J$ to decode the index $I$ and the sequence $C^n_{{I}J}$ from $B^n$.
If the rate for the index $I$ satisfies $R_c<I(B;C)$ then, { 
\renewcommand{\theenumi}{\roman{enumi}}
\begin{enumerate}
\item $\mathbb E_\mathsf{C}\big[\mathbb P[\hat{I}\neq I]\big]\rightarrow 0$ as $n\rightarrow \infty$, where $\mathbb P[\hat{I}\neq I]$ is the probability that the decoding fails for a realization of the random codebook, and \label{itm:proof2Itm1}
\item $\lim\limits_{n\rightarrow \infty} \mathbb E_\mathsf{C}\big[ \norm{\check{P}_{X^nY^nIJK}-\mathring{P}_{X^nY^nIJK}}_{{\scriptscriptstyle TV}}\big] =0.$ \label{itm:proof2Itm2}
\end{enumerate}}
\end{lemma}
\begin{proof}
We start the proof of \ref{itm:proof2Itm1}{)} by calculating the average probability of error, averaged over all codewords in the codebook and averaged over all random codebook realizations.
 \begin{align}
 \mathbb E_\mathsf{C}\big[\mathbb P[\hat{I}\neq I]\big]&= \sum_{\mathsf{C}} P_{\mathsf C}(\mathsf c) \mathbb P[\hat{I}\neq I] \notag\\
 &= \sum_{\mathsf{C}} P_{\mathsf C}(\mathsf c) \sum_{i,j,k} \frac{1}{2^{nR}}\mathbb P\Big[\hat{I}\neq I \Big| \substack{I=i\\ J=j\\K=k}\Big] \notag\\
 &=\sum_{i,j,k} \frac{1}{2^{nR}} \sum_{\mathsf{C}} P_{\mathsf C}(\mathsf c) \mathbb P\Big[\hat{I}\neq I \Big| \substack{I=i\\ J=j\\K=k}\Big] \notag\\
 &\stackrel{(a)}{=}\mathbb P\Big[\hat{I}\neq I \Big| \substack{I=1\\ J=1\\K=1}\Big] , \label{eq:DecodProof}
\end{align}
where  in (a) we have used the fact that the conditional probability of error is independent of the triple of indices due to the i.i.d.~nature of the codebook construction. Also, due to the random construction and the properties of jointly typical set, we have 
\begin{align*}
\mathbb E_{\mathsf C} [\mathds{1} \big((A_{111}^n,B^n,C^n_{11})\in {\cal T}_\epsilon^n(P_{ABC})\big)] \xrightarrow{n\rightarrow\infty} 1.
\end{align*}

We now continue the proof by constructing the sets for each $j$ and $ b^n\in\mathcal B^n$ that Node $Y$ will use to identify the transmitted index:
	\begin{align*}
	&\hat{S}_{j,b^n,\mathsf c}\triangleq  \{i: (b^n,c^n_{ij}) \in {\cal T}_\epsilon^n(P_{BC})\}.
	\end{align*}
\noindent The set $\hat{S}_{j,b^n,\mathsf c}$ consists of indices $i\in I$ such that for a given common randomness index $J=j$ and channel realization $B^n=b^n$, the sequences $(b^n,c_{ij}^n)$ are jointly-typical. Assuming $(i,j,k)=(1,1,1)$ was realized, and if $\hat{S}_{1,b^n,\mathsf c}=\{1\}$, then the decoding will be successful. The probability of this event is divided into two steps as follows:

\noindent\textbullet\; First, assuming $(i,j,k)=(1,1,1)$ was realized, for successful decoding, $1$ must be an element of $\hat{S}_{J,B^n,\mathsf c}$. The probability of this event can be bounded as follows.
\begin{align*} 
		\mathbb E_\mathsf C \Big[\mathbb P&\Big[I \in \hat{S}_{J,B^n,\mathsf C}\Big|\substack{I=1\\J=1\\{K=1}}\Big]\Big]\\
		&=\sum_{a^n,b^n,c^n} \Big(P_{C}^{\otimes n}(c^n) P_{A|C}^{\otimes n}(a^n|c^n) P_{B|A}^{\otimes n}(b^n|a^n) \\
		&\hspace{1.0in} \times \mathds{1}\big((c^n,b^n)\in {\cal T}_\epsilon^n(P_{BC})\big)\Big)\\
		&\stackrel{}{=}\sum_{b^n,c^n} P_{BC}^{\otimes n}(b^n,c^n)\mathds{1}\big((b^n\!,c^n)\!\in\! {\cal T}_\epsilon^n(P_{BC})\big)\\
		&\stackrel{(a)}{\geq} 1-\delta(\epsilon) \xrightarrow{n\rightarrow\infty} 1,
\end{align*} 
where (a) follows from the properties of jointly typical sets.\\
		
\noindent\textbullet\; Next, assuming again that $(i,j,k)=(1,1,1)$ was realized, for successful decoding no index greater than or equal to $2$ must be an element of $\hat{S}_{J,B^n,\mathsf c}$. The probability of this event can be bounded as follows:
	\begin{align*} 
	&\mathbb E_\mathsf C \mathbb P\Big [\hat{S}_{J,B^n,\mathsf C} \cap \{2,\dots,2^{nR_c}\}\!=\! \emptyset\Big|\substack{I=1\\J=1\\K=1}\Big]\\
		&\qquad=1- \sum_{i' \neq 1} \mathbb{E}_\mathsf C\mathbb P\Big [i' \in \hat{S}_{J,B^n,\mathsf C}\Big|\substack{I=1\\J=1\\K=1}\Big] \\ 
		&\qquad=1- \sum_{i' \neq 1} \mathbb P[(C^n_{i'1},B^n)\in {\cal T}_\epsilon^n(P_{BC})]\\
		&\qquad\stackrel{(a)}{\geq}1- \sum_{i' \neq 1} 2^{-n(I(B;C)-\delta(\epsilon))}\\
		&\qquad=1-(2^{nR_c}-1) 2^{-n(I(B;C)-\delta(\epsilon))}\\
		&\qquad=1- 2^{-n(I(B;C)-R_c-\delta(\epsilon))} + 2^{-nI(B;C)}\\
		&\qquad\stackrel{(b)}{\geq} 1-\delta(\epsilon) \xrightarrow{n\rightarrow\infty} 1,
		\end{align*}
where
		(a) follows from the packing lemma \cite{elGamal2011NITetwork}, and
		(b) results if $R_c <I(B;C)-\delta(\epsilon)$. 
		
Then from \eqref{eq:DecodProof}, the claim in \ref{itm:proof2Itm1}{)} follows as given by		
	\begin{align*}
		&\mathbb E_\mathsf{C}\big[\mathbb P[\hat{I}\neq I]\big] = \mathbb E_\mathsf C\mathbb P\Big[\hat{I}\neq I\Big|\substack{I=1\\J=1\\K=1}\Big] \\
		&\qquad \leq \left(\mathbb E_\mathsf C \mathbb P\Big[I \notin \hat{S}_{J,B^n,\mathsf C}\Big|\substack{I=1\\J=1\\{K=1}}\Big]\right.\\
		&\qquad\qquad\qquad+\left. \mathbb E_\mathsf C \mathbb P\Big [\hat{S}_{J,B^n,\mathsf C} \cap \{2,\dots,2^{nR_c}\}\!\neq\! \emptyset\Big|\substack{I=1\\J=1\\K=1}\Big]\right) \\
		&\qquad\xrightarrow{n\rightarrow\infty} 0
	\end{align*}
Finally, the proof of \ref{itm:proof2Itm2}{)} follows in a straightforward manner. If the previous two conditions are met, then ${\mathbb E_\mathsf{C}[\mathbb P[\hat{I}\neq I]]\rightarrow 0}$ and  $\mathbb E_\mathsf{C} [{P_{\hat{I}|B^nJ}(\hat{i}|b^n,j)]\rightarrow\delta_{I\hat{I}}}$, 
where $\delta_{I\hat{I}}$ denotes the Kronecker delta. Consequently, from \eqref{eqn-codeinducedAct} and \eqref{eqn-codeinduced} \begin{align}\lim\limits_{n\rightarrow \infty}  \mathbb E_\mathsf{C}\big[ \norm{\check{P}_{X^nY^nIJK}-\mathring{P}_{X^nY^nIJK}}_{{\scriptscriptstyle TV}}\big]  =0.\end{align}

\end{proof}

\subsection{Independence constraint}
\label{sec:independence}

We complete modifying the allied structure to mimic the original problem
with a final step. By assumption, we have a natural independence between the
action sequence $X^n$ and the common randomness $J$. As a result, the joint
distribution over $X^n$ and $J$ in the original problem is a product of the
marginal distributions $P^{\otimes n}_{X}$ and $P_J$. To mimic this
behavior in the scheme for the allied problem, in Lemma~\ref{Lma:Secrecy} we artificially enforce independence by ensuring that the mutual information between $X^n$ and $J$ vanishes.

\begin{lemma}[Independence constraint]
	\label{Lma:Secrecy}
	Consider the scheme for the allied problem given in Fig.~\ref{fig:StrongCoordinationAllied}. Both $	I(J;X^n)\rightarrow 0$ and $\mathbb E_\mathsf{C} \big[||\check{P}_{X^nJ}-P^{\otimes n}_{X}P_J||_{{\scriptscriptstyle TV}} \big] \rightarrow 0$ as $n\rightarrow \infty$ 
	if the code rates satisfy 
	\begin{align}
    R_a+R_c &> I(X;AC),\label{eqn-Indep1}\\
    R_c &> I(X;C).\label{eqn-Indep2}
    \end{align}
\end{lemma}

The proof of Lemma~\ref{Lma:Secrecy} builds on the results of
Section~\ref{sec:DecoCons} and the proof of
Lemma~\ref{Lma:Reslv} of Section \ref{sec:ChResCons}, resulting in
\begin{align} \label{eq:indepProof} 
\mathbb E_\mathsf{C} \big[||\check{P}_{X^nJ}&-P^{\otimes n}_{X}P_J||_{\scriptscriptstyle TV}\big]  
\leq \mathbb E_\mathsf{C} \Big[\sqrt{2\mathbb{D}(\check{P}_{X^nJ}||P^{\otimes n}_{X}P_J)}\;\Big] \notag\\
&\leq \sqrt{2\mathbb E_\mathsf{C} \big[\mathbb{D}(\check{P}_{X^nJ}||P^{\otimes n}_{X}P_J)\big]}\mathop{\longrightarrow}^{n\rightarrow \infty} 0.
\end{align}

\begin{remark}
	Given $\epsilon>0$, $R_a$, $R_c$ meeting \eqref{eqn-Indep1} and \eqref{eqn-Indep2}, it follows from \eqref{eq:indepProof} that there exist an $n\in\mathbb N$ and a random codebook realization for which the code-induced pmf between the common randomness $J$ and the actions of Node $X$ satisfies
	\begin{align} 
	||\check{P}_{X^nJ}-P^{\otimes n}_{X}P_J||_{{\scriptscriptstyle TV}}  <\epsilon.
	\end{align}
\end{remark}

In the original problem of Fig.~\ref{fig:P2PCoordination}, the input action sequence $X^n$ and the index $J$ from the common randomness source are available and the $A$- and $C$-codewords are to be selected. Now, to devise a scheme for the strong coordination problem, we proceed as follows. We let Node $X$ choose indices $I$ and $K$ (and, consequently, the $A$- and $C$-codewords) from the realized $X^n$ and $J$ using the conditional distribution $\mathring P_{IK|X^nJ}$. The joint pmf of the actions and the indices is then given by
\begin{align}
\hat{P}_{X^nY^nIJK} \triangleq P_X^{\otimes n} P_J \mathring P_{IK|X^nJ} \mathring P_{Y^n| IJK}. \label{eqn-randomcoordscheme}
\end{align}

\noindent Finally, we can argue that
\begin{align}
\lim_{n\rightarrow \infty} \mathbb E_\mathsf{C} [\norm{\hat{P}_{X^nY^n}- P_{XY}^{\otimes n}}_{{\scriptscriptstyle TV}} ] =0, 
\end{align} 
since the total variation between the marginal pmf $\hat{P}_{X^nY^n}$ and the design pmf $P_{XY}^{\otimes n}$ can be bounded as 
\begin{equation*}
\begin{split} 
&\norm{\hat{P}_{X^nY^n}-P_{XY}^{\otimes n}}_{\scriptscriptstyle TV}\\
&\stackrel{(a)}{\leq} \norm{\hat{P}_{X^nY^n} -\mathring P_{X^nY^n}}_{\scriptscriptstyle TV}+ \norm{\mathring P_{X^nY^n} - \check{P}_{X^nY^n}}_{\scriptscriptstyle TV} \\
& \qquad+ \norm{\check{P}_{X^nY^n} -P_{XY}^{\otimes n}}_{\scriptscriptstyle TV} \notag\\  
&\stackrel{(b)}{\leq} \norm{\hat{P}_{X^nY^nIJK}-\check{P}_{X^nJ}\mathring P_{IKY^n|X^n,J}}_{\scriptscriptstyle TV}\\
& \qquad + \norm{\check{P}_{X^nY^nIJK}-\mathring{P}_{X^nY^nIJK}}_{{\scriptscriptstyle TV}}  
+ \norm{\check{P}_{X^nY^n}\!-\!P_{XY}^{\otimes n}}_{\scriptscriptstyle TV}\notag\\
&\stackrel{(c)}{=}\norm{P_X^{\otimes n} P_J -\check{P}_{X^nJ}}_{\scriptscriptstyle TV}+ \norm{\check{P}_{X^nY^nIJK}-\mathring{P}_{X^nY^nIJK}}_{{\scriptscriptstyle TV}}\\  
& \qquad + \norm{\check{P}_{X^nY^n} -P_{XY}^{\otimes n}}_{\scriptscriptstyle TV} 
\end{split}
\end{equation*}  
\noindent where
(a) follows from the triangle inequality;
(b) follows from \eqref{eqn-codeinducedAct}, \eqref{eqn-codeinduced}, \eqref{eqn-randomcoordscheme} and \cite[Lemma V.1]{Cuff13}; (c) follows from \cite[Lemma V.2]{Cuff13}.
The terms in the RHS of (c) can be made vanishingly small provided the resolvability, decodability, and independence conditions are met. Thus, we are guaranteed that by meeting the five conditions of Lemmas~\ref{Lma:Reslv}-\ref{Lma:Secrecy}, the scheme defined by \eqref{eqn-randomcoordscheme} achieves strong coordination between Nodes $X$ and $Y$ by communicating over the DMC $P_{B|A}$. Note that since the operation at Nodes $X$ and $Y$ amount to an index selection according to $\mathring P_{IK|X^nJ}$, and a generation of $Y^n$ using the DMC $P_{Y|BC}$, both operations are randomized. The last step is to derandomize the operations at Nodes $X$ and $Y$ by viewing the corresponding local randomness as the source of randomness in these operations. This is detailed next.

\subsection{Local randomness rates}
At Node $X$, local randomness is employed to randomize the selection of indices $(I,K)$ by synthesizing the channel $\mathring P_{IK|X^nJ}$ whereas Node $Y$ utilizes its local randomness  to generate the action sequence $Y^n$ by simulating the channel $P_{Y|BC}$. Using the arguments in~\cite{SCMLN:2015}, we can argue that for any given realization of $J$, the minimum rate of local randomness required for the probabilistic selection of indices $(I,K)$ can be derived by quantifying the number of $A$ and $C$ codewords (equivalently the pair of indices $I,K$) jointly typical with $X^n$. Quantifying the list size as in~\cite{SCMLN:2015} yields $\rho_1 \geq R_a+R_c-I(X;AC)$. At Node $Y$,
the necessary local randomness for the generation of the action sequence is bounded by the channel simulation rate of DMC $P_{Y|BC}$~\cite{Steinberg-Verdu-IT-1994}. Thus,  
$\rho_2 \geq H(Y|BC)$.

Moreover, one can always view a part of the common randomness as local randomness, which then allows us to incorporate the rate-transfer arguments given in \cite[Lemma 2]{SCMLN:2015}. Combining the rate-transfer argument with the constraints in Lemmas \ref{Lma:Reslv}-\ref{Lma:Secrecy}, we obtain
following inner bound to the strong coordination capacity region.

\begin{theorem}\label{Thm:JointCRR}
	A tuple $(R_o, \rho_1,\rho_2)$ is achievable 
	for the strong noisy communication setup in
	Fig.~\ref{fig:P2PCoordination} if for some $R_a,R_c,\delta_1,\delta_2\geq0$,
	\begin{subequations}\label{equ:JSRR}
		\begin{align} 
		R_a+R_o+R_c & {\;>\;} I(XY;AC)+\delta_1 +\delta_2, \label{equ:JSRR1}\\
		R_o+R_c &{\;>\;} I(XY;C)+\delta_1 +\delta_2,\label{equ:JSRR2}\\
		R_a+R_c &{\;>\;} I(X;AC),\label{equ:JSRR3}\\
		R_c &{\;>\;} I(X;C),\label{equ:JSRR4}\\
		R_c &<I(B;C),\label{equ:JSRR5}\\
		\rho_1 &{\;>\;} R_a+R_c-I(X;AC)-\delta_1,\label{equ:JSRR6}\\
		\rho_2 &{\;>\;} H(Y|BC)-\delta_2.\label{equ:JSRR7} 
		\end{align}
	\end{subequations}
\end{theorem}

\section{Separate Coordination-Channel Coding Scheme with Randomness Extraction}\label{sec:SepCoorRE}
As a basis for comparison, we will now introduce a separation-based scheme that involves randomness extraction. We first use a $(2^{nR_c},2^{nR_o},n)$ noiseless coordination code with the codebook $\mathscr{U}$ to generate a message $I$ of rate $R_c$. Such a code exists if and only if the rates $R_o, R_c$ satisfy \cite{cuff2010coordination}
\begin{equation}
R_c+R_o \geq I(XY;U),\;R_c \geq I(X;U).\notag
\end{equation}
This coordination message $I$ is then communicated over the noisy channel using
a rate-$R_a$ channel code over $m$ channel uses with codebook $\mathscr{A}$. Hence, $R_c= \lambda
R_a$, where $\lambda=m/n$. The probability of decoding error can be made vanishingly small if $R_a <I(A;B)$. Then,
from the decoder output $\hat{I}$ and the common randomness message $J$ we
reconstruct the coordination sequence $U^n$ and pass it though a test
channel $P_{Y|U}$ to generate the action sequence at Node $Y$. Note that this
separation scheme is constructed as a special case of the joint coordination-channel scheme of Fig.~\ref{fig:StrongCoordinationAllied} by choosing $C=U$
and $P_{AC} =P_A P_U$.

In the following, we restrict ourselves to additive-noise DMCs,
i.e.,~\begin{equation} \label{equ:additiveNoise}
B^m=A^m(I)+Z^m,
\end{equation} where $Z$ is the noise random variable drawn from
some finite field $\mathcal{Z}$, and ``$+$'' is the native addition operation in the field. To extract randomness, we exploit the additive nature of the channel to recover the realization of the channel noise from the decoded codeword. Thus, at the channel decoder output we obtain \begin{equation} \label{equ:NoiseSource}
\hat{Z}^m=B^m+A^m(\hat{I}),
\end{equation} where $B^m$ is the channel output and $A^m(\hat{I})$ the corresponding decoded channel codeword. We can then utilize a randomness extractor on $\hat{Z}^m$ to
supplement the local randomness available at Node $Y$. The following lemma provides some guarantees with respect to the randomness extraction stage.  

\begin{lemma}
	\label{lem:rand}
	Consider the separation based scheme over a finite-field additive DMC. If $R_a<I(A;B)$ and ${m,n\rightarrow \infty}$ with $\frac{m}{n} = \lambda$, the following hold:
	\renewcommand{\theenumi}{\roman{enumi}} {
	\begin{enumerate}
	 \item ${\mathbb P [ Z^m \neq \hat Z^m ] \rightarrow 0},$ \label{itm:proofItm1}
	 \item ${\frac{1}{m}H(\hat Z^m) \rightarrow H(Z)},$ and \label{itm:proofItm2}
	 \item ${I(\hat Z^m; I \hat I) \rightarrow 0}$. \label{itm:proofItm3}
	\end{enumerate} }
\end{lemma}
\begin{proof}
  Let $P_e$ be the probability of decoding error (i.e.,~
  $P_{I_e}=\mathbb P[I\neq \hat{I}]$ and $P_{Z_e}=\mathbb P[Z^m \neq
  \hat{Z}^m]$). 
  We first show the claim in~\ref{itm:proofItm1}{)}. From the channel coding theorem we obtain that
  $P_{I_e} \leq 2^{-n\varepsilon}$. Consequently, from~\eqref{equ:additiveNoise} and~\eqref{equ:NoiseSource} $\mathbb P[Z^m \neq
  \hat{Z}^m]$ will follow directly as $P_{Z_e} \leq
  2^{-m\varepsilon}$.\\
\noindent Then, the claim in~\ref{itm:proofItm2}{)} is shown as follows 
\begin{equation*}
\begin{split}
 H(\hat{Z}^m)& \stackrel{(a)}{\leq} H(Z^m)+ H(\hat{Z}^m|Z^m)\\
			&\stackrel{(b)}{\leq} mH(Z) + h_2(P_{Z_e})+ P_{Z_e}m\log|{\cal Z}|\\
{\textstyle \frac{1}{m}}H(\hat{Z}^m)&\leq H(Z) + {\textstyle \frac{1}{m}} h_2(P_{Z_e})+ P_{Z_e}\log|{\cal Z}|\\
{\textstyle \frac{1}{m}}H(\hat{Z}^m)& \xrightarrow{P_{Z_e}\rightarrow 0} H(Z) \\				
\end{split}
\end{equation*}
where
(a) follows from the chain rule of entropy; (b) follows from Fano's inequality and the fact that $Z^m\sim P_{Z}^{\otimes n}$;\\ 
\noindent Finally, the claim in~\ref{itm:proofItm3}{)} is shown by the following chain of inequalities:
\begin{equation*}\label{equ:SepSchRE}
\begin{split}
I(\hat{Z}^m;I\hat{I}) & \leq I(Z^m\hat{Z}^m;I\hat{I})\\
& \leq I(Z^m \hat{Z}^m;I)+H(\hat{I}|I)\\
& = H(\hat{Z}^m|Z^m)-H(\hat{Z}^m|Z^mI)+ H(\hat{I}|I)\\
& \leq H(\hat{Z}^m|Z^m)+ H(\hat{I}|I)\\
& \stackrel{(a)}{\leq} h_2(P_{Z_e})+ P_{Z_e}m\log|{\cal Z}| + h_2(P_{I_e})+P_{I_e}nR_c\\
& \stackrel{(b)}{\leq}\epsilon 
\end{split}
\end{equation*}
where
(a) follows from Fano's inequality; (b) follows from 
{$P_{I_e} \leq 2^{-n\varepsilon}$}, {$P_{Z_e} \leq 2^{-m\varepsilon}$} and $\epsilon,\varepsilon\rightarrow 0$ as $n,m\rightarrow \infty$ respectively.	
\end{proof}

Now, similar to the joint scheme, we can quantify the local randomness at both nodes, apply the rate transfer lemma \cite[Lemma 2]{SCMLN:2015}, and set $\lambda=1$ to facilitate comparison with the joint scheme from Section \ref{sec:JointScheme}. The following theorem then describes an inner bound to the strong coordination region using the separate-based scheme with randomness extraction. 
\begin{theorem}\label{Thm:SepSchCRR}
	There exists an achievable separation based coordination-channel coding scheme for the strong setup in Fig~\ref{fig:P2PCoordination} such that (\ref{eq:StrngCoorCondtion}) is satisfied for $\delta_1\geq0, \delta_2\geq0$ if
\begin{subequations}\label{equ:SSRR}
		\begin{align}
		R_c+R_o &\geq I(XY;U) +\delta_1 +\delta_2, \label{equ:SSRR1}\\ 
		R_c &\geq I(X;U),\label{equ:SSRR2}\\
		R_c &<I(A;B),\label{equ:SSRR3}\\ 
		\rho_1 &\geq R_c-I(X;U)-\delta_1, \label{equ:SSRR4}\\
		\rho_2 &\geq \max \big(0,H(Y|U)-H(Z)\big)-\delta_2. \label{equ:SSRR5}
		\end{align}
	\end{subequations}	
\end{theorem}

The proof follows in a straightforward way from the proofs of both
Theorem~\ref{Thm:JointCRR} and Lemma~\ref{lem:rand} and is therefore omitted. 

\section{Example} \label{sec:example}
In the following, we compare the performance of the joint scheme in
Section~\ref{sec:JointScheme} and the separation-based scheme in
Section~\ref{sec:SepCoorRE}  using a simple
example. Specifically, we let $X$ be a Bernoulli-$\frac{1}{2}$ source,
the 
communication channel $P_{B|A}$ be a  binary symmetric channel with crossover probability $p_o$
(BSC($p_o$)), and the conditional distribution $P_{Y|X}$ be a BSC($p$).

\subsection{Basic separation scheme with randomness extraction}	
To derive the rate constraints for the basic separation scheme, we consider  $X-U-Y$ with $U\!\sim \!\mathrm{Bernoulli}-\frac{1}{2}$ (which is known to be optimal \cite{Cuff13}),  $P_{U|X}=\mathrm{BSC}(p_1)$, and $P_{Y|U}=\mathrm{BSC}(p_2)$,  $p_2 \in [0,p]$, $p_1=\dfrac{p-p_2}{1-2p_2}$.
Using this to obtain the mutual information terms in
Theorem~\ref{Thm:SepSchCRR}, we get
\begin{subequations}
	\begin{align}
	&I(X;U)=1-h_2(p_1),\; I(A;B)=1-h_2(p_o),\label{eq:sep_conda}\\
	&I(XY;U)=1+ h_2(p)-h_2(p_1)-h_2(p_2),\\
	&\text{and } H(Y|U)= h_2(p_2).\label{eq:sep_condc}
	\end{align}
\end{subequations}

After a round of Fourier-Motzkin elimination by using \eqref{eq:sep_conda}-\eqref{eq:sep_condc} in Theorem~\ref{Thm:SepSchCRR}, we obtain the following constraints for the achievable region using the separation-based scheme with randomness extraction:
\begin{subequations}\label{eq:SepSchemObjective}
	\begin{align}
	R_o+\rho_1+\rho_2&\geq h_2(p)-\min\big(h_2(p_2), h_2(p_o)\big),\!\label{eq:SepSchemObjective1}\\		
	h_2(p_1)&\ge h_2(p_o) \label{eq:SepSchemObjective2}\\
	R_c&\geq 1-h_2(p_1).\label{eq:SepSchemObjective3}
	\end{align}
\end{subequations} 
Note that \eqref{eq:SepSchemObjective1} presents the achievable sum rate
constraint for the total required randomness in the system.

\subsection{Joint scheme}
The rate constraints for the joint scheme are constructed in two
stages. First, we derive the scheme for the codebook
cardinalities ${|{\cal A}|=2}$ and ${|{\cal C}|=2}$, an extension to larger
$|{\cal C}|$ is straightforward but more tedious (see Figs.~\ref{fig:MinSumRateComparsion} and
\ref{fig:CommunicationRateComparision})\footnote{Note that these
	cardinalities are not optimal. They are, however, analytically feasible and provide a good intuition about the performance of the
	scheme.}. The joint scheme correlates the codebooks while ensuring that the
decodability constraint \eqref{equ:JSRR5} is satisfied. To get the best tradeoff, we find
the joint distribution  $P_{AC}$ that
maximizes $I(B;C)$. For ${|\mathcal{C}|=2}$ this is simply given by
${P_{A|C}(a|c)=\delta_{ac}}$. 
Then, the distribution $P_{X}(x)P_{CA|X}(c,a|x)P_{B|A}(b|a)P_{Y|BC}(y|b,c)$
that produces the boundary of the strong coordination region for the joint
scheme is formed by cascading two BSCs and another symmetric  channel,
yielding the Markov chain ${X-(C,A)-(C,B)-Y}$, with the channel transition matrices
\begin{align}
P_{CA|X}&=\left[ \begin{matrix}
1-p_1 &0  & 0 & p_1 \\
p_1  &0    & 0 & 1-p_1 
\end{matrix}\right],\\
P_{CB|CA}&=\left[\begin{matrix}
1-p_o  &p_o  & 0   & 0\\
0   &0    & p_o & 1-p_o 
\end{matrix}\right], \\
P_{Y|CB}&=\left[\begin{matrix}
1-\alpha  &  1-\beta & \beta & \alpha \\
\alpha  & \beta & 1-\beta & 1-\alpha
\end{matrix}\right]^T
\end{align} 
for some $\alpha,\beta \in [0,1].$\\

Then, the mutual information terms  in
Theorem~\ref{Thm:JointCRR}   can be 
expressed with $p_2\triangleq (1-p_o)\alpha + p_o\beta$ as
\begin{align*}
I(X;AC) &=I(X;C) = 1-h_2(p_1),&\\
I(XY;AC) &=I(XY;C)\\
&=1+ h_2(p)-h_2(p_1)-h_2(p_2),&\\
I(B;C) &= 1-h_2(p_o),\text{ and }\\
H(Y|BC)&= p_oh_2(\beta)\!+\!(1\!-\!p_o)h_2(\alpha).
\end{align*}
To find the minimum achievable sum rate we first perform Fourier-Motzkin
elimination on the rate constraints in Theorem~\ref{Thm:JointCRR} and then
minimize the information terms with respect to the parameters $p_2$,
$\alpha$, and $\beta$ as follows:
\begin{align}
&\!R_o\!+\!\rho_1\!+\!\rho_2\!=\!\!\min_{p_2,\alpha,\beta}\!\!
\big(h_2(p)\!-\!h_2(p_2)\!+\!(1\!-\!p_o)h_2(\alpha)\!+\!p_oh_2(\beta)\big)\notag\\
&\qquad\qquad\quad\,\,\, \text{subject to  }\begin{array}{rcl} h_2(p_1)&>&h_2(p_o),\\R_c&\geq& 1-h_2(p_1),\\p&=&p_1-2p_1p_2+p_2.\end{array}\label{eq:JointSchemObjective}
\end{align}

\subsection{Numerical  results}
 
\begin{figure}[t!]
	\centering 
	\includegraphics[scale=0.4]{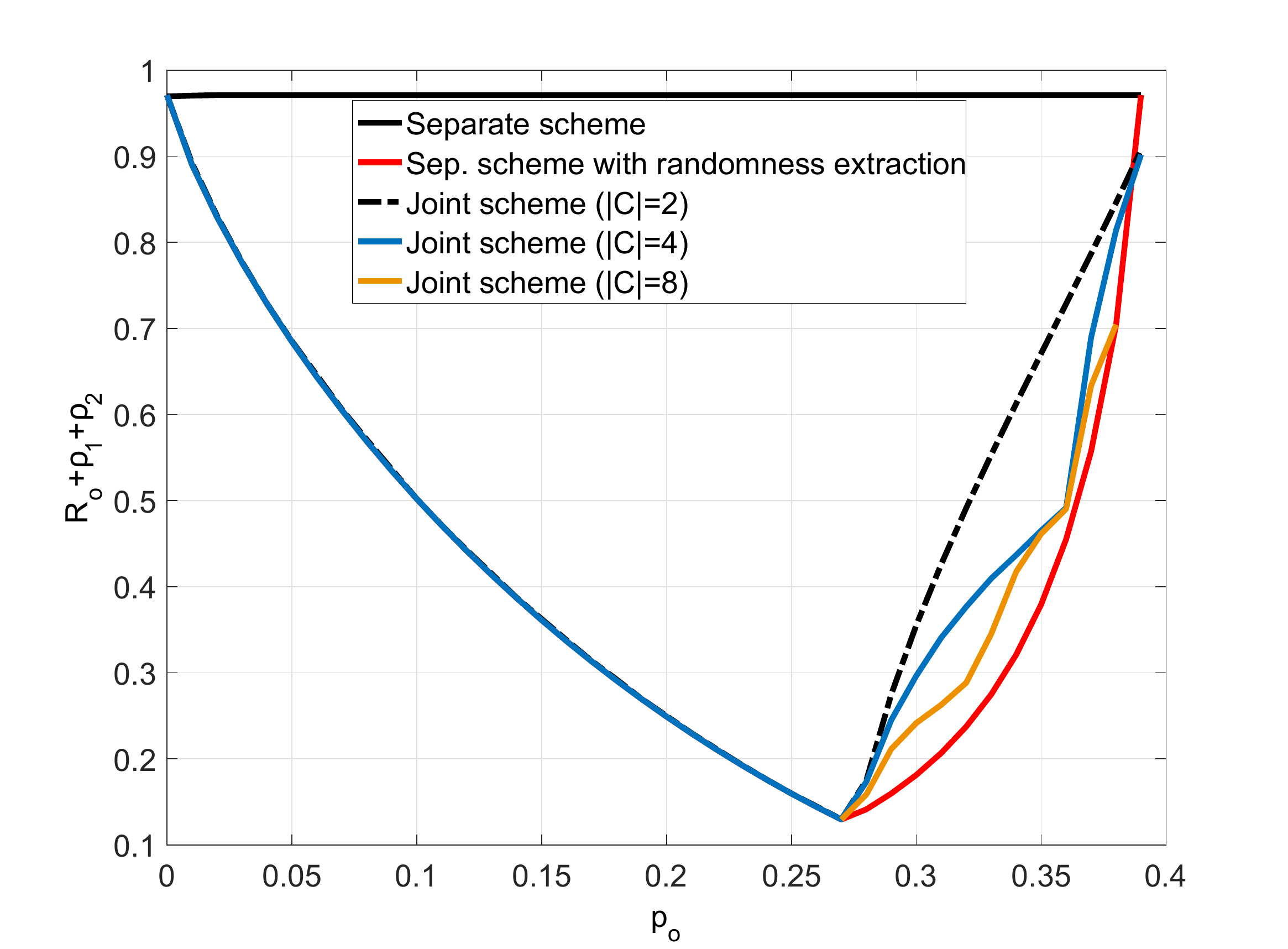}
	\vspace{-2ex}
	\caption{Randomness sum rate vs.~BSC crossover probability $p_0$.}
	\vspace{-2ex}
	\label{fig:MinSumRateComparsion}
\end{figure}

Fig.~\ref{fig:MinSumRateComparsion} presents a comparison between the
minimum randomness sum rate $R_o+\rho_1+\rho_2$ required to achieve
coordination using the joint and the separate scheme with randomness
extraction when the communication channel is given by BSC($p_o$). The target
distribution is set as $p_{Y|X}\!=\!\mathrm{BSC}(0.4)$. The rates for the
joint scheme are obtained by solving the optimization problem in
\eqref{eq:JointSchemObjective}.  Similar results are obtained for the joint
scheme with $|{\cal C}|> 2$. For the separate scheme we choose $p_2$ such
that {$h_2(p_1)=h_2(p_0)$} to maximize the amount of extracted
randomness. We also include the performance of the separate scheme without
randomness extraction. As can be seen from
Fig.~\ref{fig:MinSumRateComparsion}, both the joint scheme and the separate
scheme with randomness extraction provide the same sum rate
$R_o\!+\!\rho_1\!+\!\rho_2$ for $p_o\!\leq\! p'_o$ where
$p'_o\!\triangleq\!\frac{1-\sqrt{1-2p}}{2}$. We also observe that for noisy
channels the joint scheme approaches the performance of the separate scheme
when the cardinality of $C$ is increased. In this regime, we let $p_2=p_0$
such that $h_2(p_2)=h_2(p_0)$ in order to  maximize the amount of extracted
randomness. This is done by selecting $\alpha=0$ and $\beta=1$ associated
with $P_{Y|BC}$. However, it can be easily shown that for $p_0>p_0'$ this
does not ensure a target distribution of $P_{XY}^{\otimes n}$
anymore. Therefore, the optimization over the parameters $\alpha$ and
$\beta$ now results in a larger sum rate $R_o\!+\!\rho_1\!+\!\rho_2$ as can be seen 
from Fig.~\ref{fig:MinSumRateComparsion}. As $p_o$ increases further, the required total 
randomness of the joint scheme approaches the one for the basic separate scheme again.

Fig.~\ref{fig:CommunicationRateComparision} provides a comparison of the
communication rate for both schemes. Note that the joint scheme provides
significantly smaller rates than the separation scheme with randomness
extraction  for $p_o\leq p'_o$, independent
of the cardinality of $|\mathcal{C}|$.  Thus, in this regime joint
coordination-channel coding provides an advantage in terms of communication
cost and outperforms the separation-based scheme for the same amount of
randomness injected into the system. 

\begin{figure}[t!]
	\centering
	\includegraphics[scale=0.4]{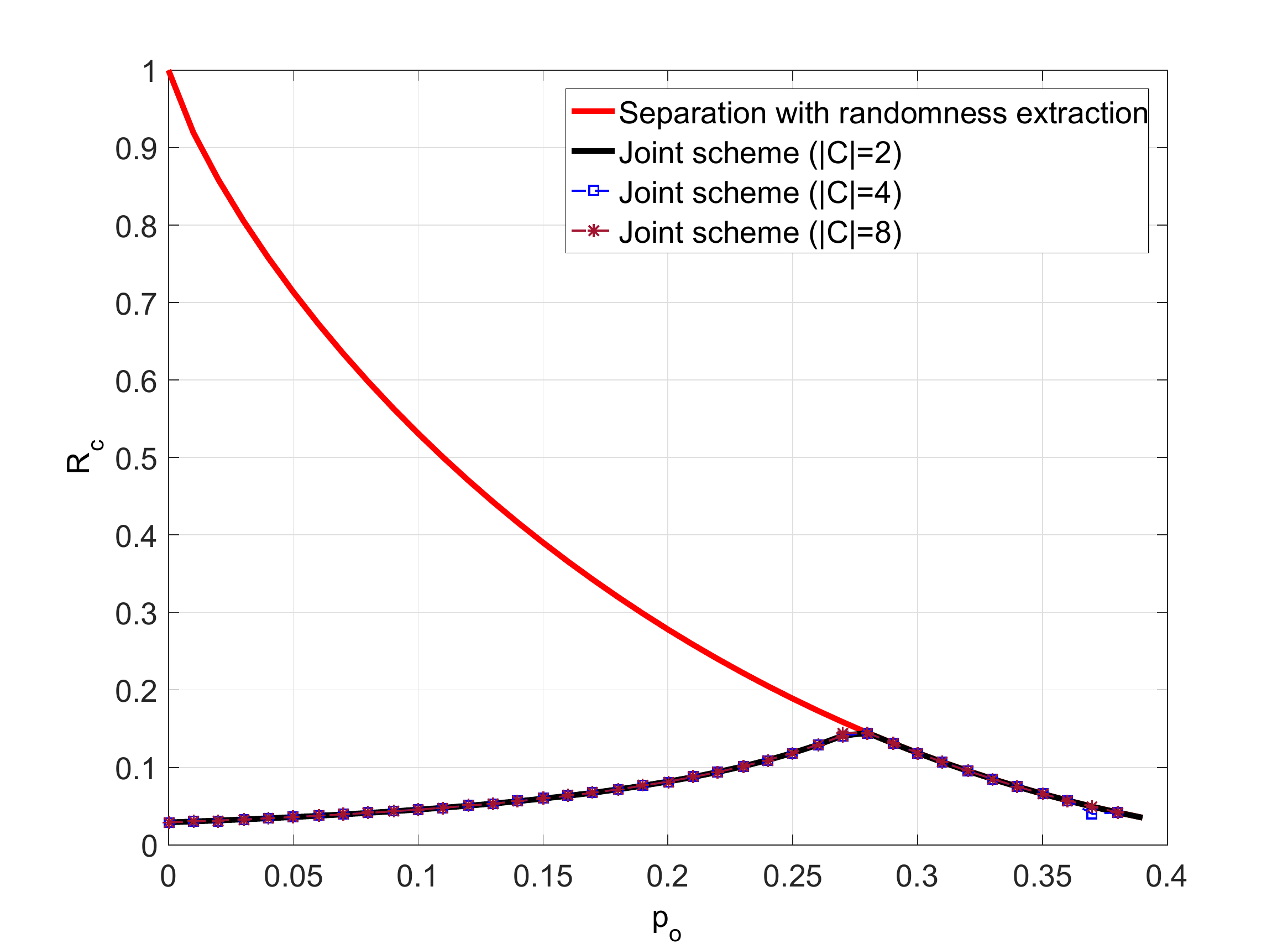}
	\vspace{-2ex}
	\caption{Communication rate vs.~BSC crossover probability $p_0$.}
	\vspace{-2ex}
	\label{fig:CommunicationRateComparision}
\end{figure}

\bibliographystyle{IEEEtran}
\bibliography{IEEEabrv,references}
\end{document}